\pdfoutput=1
\newif\ifFull
\Fulltrue
\ifFull
\documentclass[11pt]{article}
\usepackage[margin=1in]{geometry}
\else
\documentclass{sig-alternate-05-2015}
\fi
\renewcommand{\emph}[1]{\textit{\textbf{#1}}}
\usepackage{amsmath}
\usepackage{graphicx}
\usepackage[noend]{algorithmic}
\usepackage{cite}
\usepackage{url}

\newtheorem{theorem}{Theorem}%[section] uncomment to number theorems by section
\newtheorem{lemma}[theorem]{Lemma}

\ifFull
\newenvironment{proof}{\noindent{\bf Proof:}}{\hspace*{\fill}\rule{6pt}{6pt}\bigskip}
\fi

\begin{document}
% \pagestyle{plain}
% \def\thepage{\arabic{page}}

%\conferenceinfo{SPAA}{'16 June 27--30, 2016, Asilomar, CA, USA}
%\crdata{Copyright is held by the owner/author(s). \\
%Publication rights licensed to ACM. \\
%ACM 978-XXXXXXXX ...\$15.00. \\
%http://dx.doi.org/10.1145/0044}

\ifFull\else
\CopyrightYear{2016} 
\setcopyright{acmlicensed}
\conferenceinfo{SPAA'16,}{July 11--13, 2016, Pacific Grove, CA, USA.}
\isbn{978-1-4503-4210-0/16/07}\acmPrice{\$15.00}
\doi{http://dx.doi.org/10.1145/2935764.2935779} 
%Authors, replace the red X's with your assigned DOI string.
\fi

\clubpenalty=10000 
\widowpenalty = 10000
\title{Parallel Algorithms for Summing Floating-Point Numbers}

\author{Michael T.~Goodrich \\[5pt]
Dept. of Computer Science \\
University of California, Irvine \\
Irvine, CA 92697 USA \\
\texttt{goodrich@acm.org}
\and
Ahmed Eldawy \\[5pt]
Dept. of Computer Science and Engineering \\
University of California, Riverside \\
Riverside, CA 92521 USA \\
\texttt{eldawy@cs.ucr.edu}
}

\date{}

\maketitle
% \ifFull\else
% \setcounter{page}{0}
% \thispagestyle{empty}
% \fi

\begin{abstract}
The problem of {\em exactly} summing $n$ floating-point numbers
is a fundamental problem that has many applications in large-scale
simulations and computational geometry. Unfortunately, due to the
round-off error in standard floating-point operations, this problem
becomes very challenging. Moreover, all existing solutions rely on
sequential algorithms which cannot scale to the huge datasets that
need to be processed.

In this paper, we provide several efficient {\em parallel} algorithms
for summing $n$ floating point numbers, so as to produce
a faithfully rounded floating-point representation of the sum.
We present algorithms in PRAM, external-memory, and MapReduce models,
and we also provide an experimental analysis of our MapReduce algorithms,
due to their simplicity and practical efficiency.
\end{abstract}

\section{Introduction}
Floating-point numbers are a common data type for representing values
that can vary widely.
A (base-2) floating-point
representation\footnote{We take the viewpoint in this paper
   that floating-point numbers are a base-2 representation;
   nevertheless, our algorithms can easily be modified to work with
   other standard floating-point bases, such as 10.}
of a number, $x$, is a sign bit, $b$, and pair of
integers, $(b,M,E)$, such that
\[
x = (-1)^b \times (1+ 2^{-t}M) \times 2^{E-2^{l-1}-1},
\]
where $t$ is the number of bits allocated for $M$ and $l$ is the
number of bits allocated for $E$.
The value, $M$, is called the 
\emph{mantissa} or \emph{significand}
and the value $E$ is
called the \emph{exponent}.
For more background information on floating-point numbers,
please see, e.g., the excellent survey by
Goldberg~\cite{Goldberg:1991:CSK:103162.103163}.

In \emph{fixed-precision} representations, specific 
values for $t$ and $l$ are set
according to machine-architecture or established standards.
For example, in the IEEE 754 standard, single-precision
floating-point numbers use $t=23$ and $l=8$, double-precision
floating-point numbers use $t=52$ and $l=11$,
and quad-precision floating-point numbers use $t=112$ and $l=15$.

In \emph{arbitrary-precision} representations, 
the number of bits, $t$, for the mantissa, is allowed to vary, either 
to be as large as a machine's memory size or to be 
an arbitrary value set by a user. 
The number, $l$, of bits for the exponent in such
representations is typically set to be a single memory word, since a single
memory word is usually sufficient to store any memory address.
Examples of systems for processing
arbitrary-precision floating point representations
include Apfloat for Java~\cite{apfloat}, the GNU Multiple Precision (GMP)
Arithmatic Library~\cite{GMP},
the \texttt{bigfloat} type in LEDA~\cite{leda},
and
the GNU Multiple-Precision Floating-Point (MPFR) Library~\cite{mpfr,mpfr07}.

In either fixed-precision or arbitrary-precision representations, we do 
\emph{not} assume in this paper 
that floating-point numbers are normalized so that the 
most significant bit of the mantissa is $1$.

Given a set $X=\{x_1,x_2,\ldots,x_n\}$, of $n$ floating-point
numbers,
we are interested in the design and analysis of 
parallel algorithms for computing a floating point number, $S_n^*$,
that best approximates the sum
\[
S_n = \sum_{i=1}^n x_i.
\]
Moreover, we are interested in methods that are not limited
to a specific fixed-precision representation, such as IEEE 754 
double-precision.
In particular,
the specific problem we are interested in is to
compute the floating-point number $S_n^*$ that is a 
\emph{faithfully rounded}
representation of $S_n$, where
we consider the value $S_n^*$ to be faithfully rounded as long as
$S_n^*$ is either the largest floating-point number less than or
equal to $S_n$ or the smallest floating-point number greater than or
equal to $S_n$.
For other ways of rounding floating-point numbers,
please see the IEEE 754 standard.

Although the problem of accurately summing $n$ floating-point numbers
might at first seem simple, it is
surprisingly challenging, due to the roundoff error inherent
in floating-point arithmetic.
Standard floating-point addition, which we denote as ``$\oplus$,'' is
not exact, so that for two floating-point numbers, $x$ and $y$,
\[
x \oplus y = (x+y)(1+\epsilon_{x,y}); |\epsilon_{x,y}| \le |\epsilon|,
\]
where $\epsilon_{x,y}$ is a summand-specific roundoff error
and $\epsilon$ is a machine-specific worst-case 
error term~\cite{Knuth:1997:ACP:270146,doi:10.1137/S0097539798341594}.

As is well-known,
a standard way to sum $n$ numbers, given the set $X$,
is to place the numbers from $X$ in the
leaves of a binary summation tree, $T$, where each internal node, $v$, is
associated with the sum of the values in $v$'s two children,
which in this case could be the $\oplus$ operation.
Unfortunately, if the collection, $X$, of summands contains positive and
negative numbers, it is NP-complete to find the summation tree, $T$,
that utilizes only the $\oplus$ operation
and minimizes the worst-case error for the
sum~\cite{doi:10.1137/S0097539798341594}.
Even for the case when
all the numbers in $X$ are positive, the optimal $\oplus$-based summation
tree is a Huffman tree~\cite{doi:10.1137/S0097539798341594}, 
which is not necessarily efficient to compute
and evaluate in parallel.
Thus, designing efficient parallel algorithms for accurately summing $n$
floating-point numbers is a challenge.

Because of these complexity issues and the uncertainty that comes
from roundoff errors, many recent 
floating-point summation algorithms are based 
on exactly computing a rounded
sum of two floating point numbers, $x$ and $y$,
and the exact value of its error, utilizing a function,
\newcommand{\add}{{\rm AddTwo}}
\[
\add(x,y) \rightarrow (s,e_s),
\]
where $s = x\oplus y$ and $x+y = s+e_s$, with $s$ and $e_s$ being
floating point numbers in the same precision as $x$ and $y$.
Example implementations of the $\add$ function include algorithms
by Dekker~\cite{dekker}
and Knuth~\cite{Knuth:1997:ACP:270146},
both of which utilize a constant number of floating-point operations.
As with these recent approaches, in this paper we take the approach
of summing the numbers in $X$ exactly and then 
faithfully rounding this exact sum to the an appropriate
floating-point representation.

As mentioned above,
because we desire precision-independent algorithms, which can work
for either fixed-precision or arbitrary-precision representations, 
we do not take the perspective in this paper that 
floating-point numbers must 
be restricted to a specific floating-point representation, such as
IEEE 754 double-precision.
As is common in other precision-independent algorithms, however, we 
do assume that a floating-point representation is sufficiently precise so
that the number, $n$, could itself be represented exactly as a sum of 
a constant number of floating-point numbers.
A similar assumption is common in integer-based RAM and PRAM 
algorithms,\footnote{A PRAM is a synchronized 
  parallel random-access machine model~\cite{DBLP:books/aw/JaJa92},
  where memory is shared between processors, so that memory accesses are
  exclusive-read/exclusive-write (EREW), concurrent-read/exclusive-write 
  (CREW), or concurrent-read/concurrent-write (CRCW).}
and poses no restriction in practice, since it amounts to saying that
a memory address can be stored in $O(1)$ memory words.
For example, even if all the computer storage on 
earth\footnote{As of the writing of this paper, the total computer storage on
  earth is estimated to be less than $8$ zettabytes, that is, 
  $8$ trillion gigabytes.}
were interpreted
as IEEE 754 single-precision numbers and summed, we could represent 
the input size, $n$, as the (exact) 
sum of at most four single-precision floating-point numbers.

As a way of characterizing difficult problem instances, having the
potential for significant amounts of cancellation, the
\emph{condition number}, $C(X)$, for $X$, is 
defined
as follows~\cite{doi:10.1137/070710020,Zhu:2010:A9O,doi:10.1137/050645671}:
\[
C(X) = \frac{\sum_{i=1}^n |x_i|}{\left|\sum_{i=1}^n x_i \right|},
\]
which, of course, is defined only for non-zero sums.\footnote{We 
could alternatively define a condition number that 
adds to the denominator a very small value 
to avoid division by zero.}
Intuitively, problem instances with large condition numbers are more
difficult than problem instances with condition numbers close to 1.

Our approach for designing efficient parallel algorithms for summing
$n$ floating point numbers, even for difficult problem instances
with large condition numbers, is to utilize an approach 
somewhat reminiscent of the classic Fast Fourier Transform
(FFT) algorithm (e.g., see~\cite{Brigham,GoodrichADA}).
Namely, to compute the sum of $n$ floating-point 
numbers, we convert the numbers to an 
alternative representation, compute the sum of the numbers exactly in
this alternative representation, and then convert the result back to
a (faithfully-rounded) floating-point number.
In our case, the important feature of our alternative representation
is that it allows us to compute intermediate sums without 
propagating carry bits, which provides for superior parallelization.

\subsection{Previous Related Results}
The floating-point summation problem 
poses interesting challenges for parallel computation, because most existing
published exact-summation algorithms are inherently sequential.
For instance, we are not aware of any previous parallel methods
for this problem
that run in worst-case polylogarithmic time.

Neal~\cite{Neal15a} describes sequential algorithms that use a number
representation that he calls a \emph{superaccumulator} to exactly
sum $n$ floating point numbers, which he then converts 
to a faithfully-rounded floating-point number.
Unfortunately, while Neal's superaccumulator representation
reduces carry-bit propagation, it does not eliminate it, as is needed 
for highly-parallel algorithms.
A similar idea has been used in ExBLAS~\cite{CDG+14}, an open source
library for exact floating point computations.
Shewchuck~\cite{shewchuk} describes an alternative
representation for exactly representing intermediate results of
floating-point arithmetic, but his method also does not eliminate
carry-bit propagation in summations; hence, it also does not lead to
efficient parallel algorithms.
In addition to these methods, there are a number of inherently
sequential methods for exactly summing $n$ floating point numbers
using various other data structures for representing intermediate results,
including 
ExBLAS~\cite{CDG+14} and algorithms by 
Zhu and Hayes~\cite{doi:10.1137/070710020,Zhu:2010:A9O},
Demmel and Hida~\cite{doi:10.1137/S1064827502407627,demmel},
Rump {\it et al.}~\cite{doi:10.1137/050645671},
Priest~\cite{p145549},
and Malcolm~\cite{Malcolm:1971}.
Although the method of Demmel and Hida~\cite{demmel} can be
carry-free for a limited number of summands,
none of these sequential methods utilize a completely carry-free intermediate
representation suitable for efficient parallelization.
Nevertheless, Rump {\it et al.}~\cite{doi:10.1137/050645671} provide
an interesting analysis that the running
time of their sequential method 
can be estimated to depend on $n$ and a logarithm 
of the condition number and other factors.
Also, Demmel and Hida~\cite{demmel} showed that summing the numbers
in a decreasing order by exponent yields a highly accurate answer,
yet, the answer does not have to be faithfully-rounded.

\sloppy
There has also been a limited amount of
previous work on parallel algorithms for summing floating-point numbers.
Leuprecht and Oberaigner~\cite{par82}
describe parallel algorithms for summing
floating-point numbers, for instance, but
their methods only parallelize data pipelining, however, and
do not translate into efficient parallel methods with polylogarithmic 
running times or efficient algorithms in the 
external-memory\footnote{The external-memory model~\cite{Vitter:2008} 
 is a data-parallel model, where data is transferred between an internal
 memory of a specified size, $M$, and external memory in blocks of a 
 specified size, $B$.}
or 
MapReduce models.\footnote{A MapReduce algorithm is a coarse-grain
  parallel method that performs rounds that involve mapping elements to 
  keys, assigning elements by their keys to ``reducer'' processors,
  and having those processors execute on sets of elements with 
  the same key, possibly with a global sequential finalize 
  phase at the end~\cite{Dean:2008,mapsort,Karloff:2010}.}
Indeed, their methods can involve as many as $O(n)$ passes over the data.
Kadric {\it et al.}~\cite{KGD6545903} provide a parallel pipelined method 
that takes a similar approach to the algorithm of Leuprecht and Oberaigner,
while improving its convergence in practice, but 
their method nevertheless depends on
inherently sequential pipelining and iterative refinement operations.
Recently,
Demmel and Nguyen~\cite{DN6875899} present a parallel floating-point summation
method based on using a superaccumulator, 
but, like the previous sequential superaccumulator methods
cited above, their method does not
utilize a carry-free intermediate representation; hence, it has an inherently
sequential carry-propagation step as a part of its ``inner loop'' computation.

\fussy

To our knowledge, no previous algorithm for summing
floating-point numbers utilizes a carry-free intermediate representation,
although such representations are known 
for integer arithmetic (e.g., see~\cite{P46283}),
the most notable being the redundant binary representation (RBR), 
which is a positional binary representation where each position is from
the set $\{-1,0,1\}$.

\subsection{Our Results}

We show in this paper that
one can compute a faithfully-rounded sum, $S_n^*$, of $n$ 
floating-point numbers with the following performance bounds.
% Most of
% these result hold independent of whether the floating-point numbers
% are fixed-precision or arbitrary-precision numbers.

\begin{itemize}
\item
$S^*_n$ can be computed
in $O(\log n)$ time using $n$ processors in the EREW PRAM model.
This is the first such PRAM algorithm and, as we show, it is worst-case 
optimal.
\item
$S^*_n$ can be computed 
in $O(\log^2 n \log\log\log C(X))$ time using $O(n\log C(X))$ work 
in the EREW PRAM model, where $C(X)$ is the condition number of $X$.
This is the first parallel summation algorithm 
whose running time is condition-number sensitive.
\item
$S^*_n$ can be computed in external-memory
in $O({\rm sort}(n))$ I/Os, where ``sort$(n)$'' is the 
I/O complexity of sorting.\footnote{${\rm sort}(n)=O((n/B)\log_{M/B} (n/B))$,
  where $M$ is the size of internal memory and $B$ is the block size.}
\item
$S^*_n$ can be computed in external-memory
in $O({\rm scan}(n))$ I/Os, where ${\rm scan}(n)$ is 
the I/O complexity of scanning\footnote{${scan}(n)=O(n/B)$, where $B$ is 
   the block size.},
when the size of internal memory, $M$, is $\Omega(\sigma(n))$,
where $\sigma(n)$ is the size of our intermediate representation for 
the sum of $n$ floating-point numbers.
This I/O performance bound is, of course, optimal for this case.
Typically, $\sigma(n)$ is $O(\log n)$ in practice,
due to the concentration of exponent values in real-world data
(e.g., see~\cite{Isenburg2005869,LH230407}) 
or because fixed-precision floating-point
numbers have relatively small values for $l$, the number of bits used
for exponent values, compared to $n$.
\item
$S^*_n$ can be computed
with a single-round MapReduce algorithm in $O(\sigma(n/p)p+n/p)$ time
and $O(n)$ work, using $p$ processors, with high probability,
assuming $p$ is $o(n)$.
\end{itemize}
In addition, because of the simplicity of our MapReduce algorithm,
we provide an experimental analysis
of several variations of this algorithm.
We show that our MapReduce algorithm can achieve up-to 80X speedup
over the state-of-the-art sequential algorithm. It achieves a linear scalability
with both the input size and number of cores in the cluster.

\section{Our Number Representations}
To represent the exact summation of $n$ floating point numbers,
we utilize a variety of alternative number representations.
To begin,
recall that a (base-2) floating-point number, $x$, is represented as
a sign bit, $b$, and
a pair of integers, $(b,M,E)$, 
such that
\[
x = (-1)^b \times (1+ 2^{-t}M) \times 2^{E-2^{l-1}-1},
\]
where $t$ is the number of bits allocated for $M$ and $l$ is the 
number of bits allocated for $E$.
For example, for double-precision floating-point numbers in the IEEE 754
standard, $t=52$ and $l=11$.

Thus, for the problem of summing the $n$ floating-point numbers in $X$,
we could alternatively represent every floating point number, including
each partial sum, as a
fixed-point binary number consisting of a sign bit, 
$t+2^{l-1}+\lceil\log n\rceil$ 
bits to the left of the binary point, and $t+2^{l-1}$ bits to the right of
the binary point.
Using such a representation, 
we would have
no errors in summing the numbers in $X$, assuming
there are no unnecessary overflows.
In other words,
we can consider such a fixed-point representation
as a large binary integer
representation, which then has its binary point shifted.
Of course, such a representation would potentially waste a lot of 
memory,
but it is nevertheless instructive to contemplate this as a possibility,
as it motivates our other representations.

Although this fixed-point representation is potentially wasteful of memory, 
it might not actually be so bad for exactly 
representing partial sums of a reasonable numbers of summands of 
single-precision numbers.
For example, in the IEEE 754 standard, single-precision 32-bit
floating-point numbers can be represented as 256-bit fixed-point 
numbers.
Thus, 
the memory needed for an error-free fixed-point representation of
a single-precision floating-point number would
occupy the same amount of space as 8 single-precision numbers.
Nevertheless,
an additional drawback of this representation is that, in the
worst-case, there can be a lot of carry-bit propagations that occur for any
addition, which negatively impacts parallel performance.

In a \emph{superaccumulator} 
representation, as used in some previous summation 
approaches
(e.g., see~\cite{doi:10.1137/1028001,Neal15a,p145549,shewchuk,Zhu:2010:A9O}),
we instead represent a fixed-point (or floating-point) 
number as a vector, $Y$, of 
floating-point numbers, $(y_k,y_{k-1},\ldots,y_0)$, 
such that the number, $y$, represented by $Y$ is
\[
y = \sum_{i=0}^k y_i,
\]
where the numbers in $Y$ have strictly increasing exponents, 
so that $y_0$ is the
least significant number (and this summation is a true addition, with no
roundoff errors).
% Typically, in previous sequential algorithms based on
% using superaccumulators, 
% each $y_i$ is positive and there is a sign bit to represent
% the sign of $y$.

As mentioned above,
there are some interesting recent floating-point exact summation papers based
on the use of superaccumulators, but these are of limited value with respect
to parallel algorithms.
For instance, Zhu and Hayes~\cite{Zhu:2010:A9O}
present an inherently sequential algorithm that essentially involves 
adding the numbers from $X$ one at time to an initially-zero
superaccumulator and then propagating the carries to produce a
faithfully-rounded sum from the most-significant entry.
Neal~\cite{Neal15a} presents a similar algorithm and also
discusses in passing an alternative approach based on a bottom-up
parallel summation tree computation based on superaccumulators, 
but his parallel algorithm involves 
inefficiently constructing superaccumulators for every leaf, and his
superaccumulators overlap by half their bits and even then are not
carry-free.

In our case, 
we allow each $y_i$ to be positive or negative and we 
add some additional conditions so that 
superaccumulator summations are carry-free.
We do this by extending the generalized signed digit (GSD) 
integer representation~\cite{P46283} to floating-point numbers.
This is a redundant representation, so that there are multiple
ways of representing the same fixed-point number.

To simply our discussion, let us shift the binary point for $y$ so that every
$y$ is an integer instead of a fixed-point number.
This shift is made without loss of generality, since we can shift
the binary point back to its proper place after we can computed an
exact representation of the sum of numbers in $X$.

Next, following the GSD approach~\cite{P46283},
we say that a superaccumulator is \emph{$(\alpha,\beta)$-regularized} if
\[
y_i = Y_i \times R^i,
\]
for a given \emph{radix}, $R$,
and each mantissa $Y_i$ is an integer in the range $[-\alpha,\beta]$,
for $\alpha,\beta\ge 2$.
In particular, for our algorithms, for any fixed $t$,
we choose $R$ to be a power of two, $2^{t-1}>2$, 
so that each $y_i$ can be represented using a
floating-point exponent storing a multiple of $t-1$ 
(since we assume floating-point 
representations are base-2 in this paper).
For arbitrary-precision floating-point numbers, we likewise choose 
$R=2^{t_0-1}>2$,
for a fixed value $t_0\ge 2$ that is a reasonable length 
for the number of bits needed for a mantissa (proportional to our word size).
In either case, we set
\[
\alpha = \beta = R-1.
\]
This choice of the parameters, $\alpha$ and $\beta$, 
is done so that we can achieve
propagation-free carrying of components of a summation, as we show next.

Our parallel algorithm for summing two
superaccumulators, $y$ and $z$, is as follows.
First, we compute each component-wise sum of the mantissas, 
$P_i = Y_i+Z_i$. 
This sum is then reduced to
an interim mantissa sum, $W_i=P_i-C_{i+1}R$, where $C_{i+1}$ is a
\emph{signed carry bit}, i.e., $C_{i+1}\in\{-1,0,1\}$,
that is chosen
so that $W_i$ is guaranteed to be in the range $[-(\alpha-1),\,\beta-1]$.
(We show below that this is always possible.)
The computed mantissa sum is 
then performed as
$S_i = W_i + C_i$, so that the resulting collection of $s_i$ components
is $(\alpha,\beta)$-regularized and no carry-bit propagation is necessary.
As the following lemma shows,
taking this approach allows us to avoid propagating carries across
an entire superaccumulator after each addition in a summation computation,
while nevertheless representing each partial sum exactly.

\begin{lemma}
It is sufficient to choose 
$ \alpha=\beta=R-1$,
for $R>2$,
for the resulting sum of $y$ and $z$ 
to be $(\alpha,\beta)$-regularized, i.e., so that each
$S_i$ is in the range $[-\alpha,\beta]$.  
\end{lemma}
\begin{proof}
Given that $y$ and $z$ are $(\alpha,\beta)$-regularized, the
mantissa, $P_i$, is in the range $[-2\alpha,2\beta]$.
We wish to show that the result, $s=y+z$, be $(\alpha,\beta)$-regularized,
that is, that each mantissa, $S_i$, is in the range
$[-\alpha,\beta]$.
Note that if $-R+1<P_i<R-1$, then we can ensure $S_i$ is
in the range $[-\alpha,\beta]$ for any $C_i$ in $\{-1,0,1\}$
by setting $C_{i+1}=0$.
So let us consider the cases when $P_i$ is too negative or too positive.

Case 1: $P_i\ge R-1$.
Note that
\[
P_i \le 2\beta = 2R-2.
\]
Hence, we can choose $C_{i+1}=1$, so that
\[
W_i = P_i - C_{i+1}R = P_i - R \le R-2 = \beta-1.
\]
Moreover, $W_i\ge -1 \ge -(\alpha-1)$, in this case, since $\alpha\ge 2$.
Thus, 
\[
S_i = W_i + C_i \le R-2+1 \le R-1 = \beta,
\]
and $S_i\ge -\alpha$.

Case 2: $P_i\le -R+1$.
Note that
\[
P_i \ge -2\alpha = -2R+2.
\]
Hence, we can choose $C_{i+1}=-1$, so that
\[
W_i = P_i - C_{i+1}R = P_i + R \ge -R+2.
\]
Moreover, $W_i\le 1 \le \beta-1$, in this case, since $\beta\ge 2$.
Thus,
\[
S_i = W_i + C_i \ge -R+2-1 \ge -(R-1) = -\alpha,
\]
and $S_i\le \beta$.
\end{proof}

This implies that, by choosing $\alpha=\beta=R-1$,
we can guarantee
that the sum of two $(\alpha,\beta)$-regularized
superaccumulators can be done
in parallel in constant time, with each carry going to at most an
adjacent component and no further.
There is a chance, of course that the sum of two superaccumulators
indexed from $0$ to $k$ could result in a superaccumulator
indexed from $0$ to $k+1$.
But this causes no problems for us, since any mantissa can hold
$\Omega(\log n)$ bits; hence, one additional suparaccumulator component
is sufficient for holding all the adjacent-component carries during the sum of 
$n$ floating-point numbers.

An additional observation is that, since $R$ is a power of two 
in our case, we
can compute each $W_i$ and $C_i$ using simple operations involving 
the addition or subtraction
of a power of two to a number,
given $P_i$. 
Also, since we reserve an extra bit in each superaccumulator,
computing $P_i$ can be done without overflow in the standard
floating-point representation of components.

Intuitively,
each pair of consecutive superaccumulator components, $y_i$ and $y_{i+1}$, 
``{overlap}'' by one bit and each component has one additional sign bit.
As an alternative to representing each $y_i$ as a floating-point
number, then, we could 
instead represent each number, $y_i$, using an integer representation, with an
explicit overlap of one bit between each consecutive pair of numbers,
$y_i$ and $y_{i+1}$, being understood.
This would allow us to save some space for what amounts to redundant
exponents in the floating-point representations of the $y_i$'s.
For the sake of simplicity, however, 
we choose in this discussion to assume that each
$y_i$ is itself a floating-point number, with the understanding that our
algorithms could be easily modified to work for the case that the numbers in
$Y$ are integers.
In either case,
the overlap between consecutive numbers in $Y$ 
allows us to apply a lazy strategy for accumulating carry bits, 
without full propagation,
which overcomes a shortcoming in previous representations.

One important comment is in order here.
Namely, for the analysis in this paper, we are \emph{not} assuming that 
the number, $l$, of bits allocated for floating-point exponents is a fixed
constant; hence, our analysis does not assume that, for our floating-point
numbers, the size of an equivalent fixed-point representation 
or superaccumulator for this number is a fixed constant.

Thus, for small numbers of 
arbitrary-precision floating-point numbers, it is possible that
our $(\alpha,\beta)$-regularized superaccumulators may waste space.
To avoid these issues,
we can represent numbers using a format we 
are calling a ``{sparse superaccumulator}.''
Given a superaccumulator,
\[
Y=(y_k,y_{k-1},\ldots,y_0),
\]
the \emph{sparse superaccumulator} for $Y$ is the vector,
\[
Y'=(y_{i_j},y_{i_{j-1}},\ldots,y_{i_0}),
\]
consisting of all the \emph{active} indices in $Y$, for $i_0<i_1<\cdots<i_j$.
We say that a index, $i$, in a superaccumulator
is \emph{active} if $y_i$ is currently non-zero or has 
ever been non-zero in the past (when viewing a superaccumulator as
an indexed data structure).

One possible parallel algorithm for
summing two sparse superaccumulators,
$Y'=(y_{i_{j_1}},\ldots,y_{i_{0}})$
and $Z'=(z_{i_{j_2}},\ldots,z_{i_{0}})$,
is as follows.
We first merge the active indices of $Y'$ and $Z'$, and we then do a
summation of corresponding terms (possibly with carries into adjacent
components, as needed).
Note, though, that we will not propagate carries, when we use
$(\alpha,\beta)$-regularized superaccumulators.
Thus, an index in the sum is active if it was active in $Y'$ or $Z'$ or
becomes non-zero as a result of the sum.
This definition is somewhat related to the adaptive
floating-point representation of Shewchuk~\cite{shewchuk},
which introduces sparsity and adaptivity but only for 
vectors of non-overlapping floating-point numbers having arbitrary
exponents, rather than exponents that are powers of the radix $R$,
as in our sparse $(\alpha,\beta)$-regularized superaccumulator representation.
% Some of our parallel algorithms are loosely based on this approach,
% but improve upon on it in 
% multiple ways.

Furthermore,
given a sparse superaccumulator, $Y'$, and a parameter $\gamma$, we define the 
\emph{$\gamma$-truncated sparse superaccumulator} for $Y'$
to be the vector, $Y''$,
consisting of the first (most-significant) $\gamma$ entries in $Y'$.

\section{Our Fast PRAM Algorithm}
\label{sec:pram}

The
first PRAM algorithm we present is based on 
summing numbers represented using sparse $(\alpha,\beta)$-regularized
superaccumulators.
Our method runs in $O(\log n)$ time using $O(n\log n)$ work in the EREW PRAM
model, which is worst-case optimal.
The details are as follows.
\begin{enumerate}
\item
Build a binary summation tree, $T$, with $\lceil\log n\rceil$ depth,
having each leaf $i$ associated with a distinct floating-point number,
$x_i$, in $X$.
This step can be done in $O(\log n)$ time and $O(n)$ work.
\item
In parallel, for each $x_i$, convert $x_i$
into an equivalent 
$(\alpha,\beta)$-regularized superaccumulator, $x_i'$.
This step can be done in $O(1)$ time and $O(n)$ work, just by
splitting each floating-point number into $O(1)$ numbers such that
each has an exponent that is a multiple of $R$.
\item
Perform a parallel merge-sort of the $x_i'$ components, using their
exponents as keys (not their mantissas).  This creates a sorted list,
$E(v)$, for each node in $T$, consisting of the exponents found in
the subtree in $T$ rooted at $v$, as well as links for each exponent,
$e$ in $E(v)$ to its predecessors in the lists for $v$'s children and
its parent.
This step can be done in $O(\log n)$ time using $O(n\log n)$
work~\cite{doi:10.1137/0217049,Goodrich:1996:SPP:226643.226670} 
via the cascading divide-and-conquer 
technique~\cite{doi:10.1137/0218035}, because the boundary exponents
are known from the beginning.
\item
Perform a parallel prefix computation to remove duplicate exponents
in each $E(v)$.
This step can be done in $O(\log n)$ time using $O(n\log n)$
work, since the total size of all the lists is $O(n\log n)$.
\item
Using the links created in the previous step to match up
corresponding components in constant-time per level, possibly adding
new components that represent a carry bit moving into a component
that was previously not included in the sparse representation,
perform a bottom-up sparse superaccumulator summation in $T$. This
results in a sparse superaccumulator representation of the sum being
stored in the root of $T$.
This step can be done in $O(\log n)$ time using $O(n\log n)$ work.
\item
Convert the sparse $(\alpha,\beta)$-regularized superaccumulator for
the sum at the root of $T$ into a non-overlapping
superaccumulator that is
$((R/2)-1,(R/2)-1)$-regularized.
This amounts to a signed carry-bit propagation operation, which can be
done by a parallel prefix computation (based on a simple lookup
table based on whether the input carry bit is a $-1$, $0$, or $1$).
We leave the details to the interested reader.
This step can be done in $O(\log n)$ time using $O(n)$ work.
\item
Correctly round the non-overlapping
superaccumulator from the previous step
into a floating-point number.
This step amounts to locating the most significant non-zero component
of the superaccumulator and then combining that, as needed, 
with its neighboring components to create a floating-point number of 
the appropriate size, rounding the result based on the truncated bits.
This step can be done in $O(\log n)$ time using $O(n)$ work.
\end{enumerate}

This gives us the following.

\begin{theorem}
Given $n$ floating-point numbers, one can compute a 
faithfully-rounded representation
of their sum in $O(\log n)$ time using $O(n\log n)$ work (i.e., $n$ processors)
in the EREW PRAM model.
These time and processor bounds are worst-case optimal 
in the algebraic computation tree model.
\end{theorem}
\begin{proof}
We have already established the upper bounds.
For the lower bounds, we show that the set equality problem can be reduced
to the floating-point summation problem in $O(1)$ time using $n$ processors.
Suppose, then, that we are given two sets of $n$ positive numbers, $C$ and $D$,
and wish to determine if $C=D$.
Let $\tau$ be the smallest power of two greater than $\log n$.
For each element $c_i$ in $C$, create the floating-point number, 
$(-1,1,\tau c_i)$, which represents the number, $(-1)\times 2^{\tau c_i}$.
Likewise,
for each element $d_i$ in $D$, create the floating-point number, 
$(1,1,\tau d_i)$, which represents the number, $1\times 2^{\tau d_i}$.
We claim that $C=D$ if and only if the sum of these two sets 
of numbers is zero.
Clearly, if $C=D$, then there is a one-to-one matching between equal elements 
in $C$ and $D$; hence, there is a matching for 
each floating-point number, $(-1,1,\tau c_i)$, to
a floating-point number, $(1,1,\tau d_i)$, such that $c_i=d_i$.
Therefore, the sum of all of these numbers is zero in this case.
Suppose, for the ``only if'' case, that the sum of all these numbers is zero.
Note that
the exponent of any pair of floating-point numbers in our collection is 
either the same or differs by at least $\tau > \log n$.
Thus, if two such numbers are different, they will remain different
even if we multiply one of them by $n$.
Therefore, the only way that the sum of all these numbers is zero is if each
floating-point number, $(-1,1,\tau c_i)$, in this collection,
has a distinct associated 
floating-point number,
$(1,1,\tau d_i)$, in this collection, such 
that $c_i=d_i$.
Therefore, if the sum of all these numbers is zero, then $C=D$.
The lower-bounds follow, then, from the fact that summation
is a binary operator and the set equality problem
has an $\Omega(n\log n)$ lower bound in the algebraic computation
tree model~\cite{Ben-Or:1983}.
\end{proof}

To our knowledge, this is the first PRAM method that achieves $O(\log n)$
time and an amount of work that is worst-case optimal.
We note, however, that the above lower bound holds only for floating-point
representations where exponents are represented with $\Omega(\log n)$ bits.

\section{Our Condition-Number Sensitive PRAM Algorithm}

In our fast PRAM algorithm, we showed that it is possible to sum
two sparse superaccumulators in $O(1)$ time using $n$ processors
in the EREW PRAM model, given a sorted merge of the exponents of the
components of each superaccumulator.
This result also gives us the following.

\begin{lemma}
\label{lem:merge}
Given two truncated
$(\alpha,\beta)$-regularized sparse superaccumulators, $Y_1$ and $Y_2$,
of combined size $m$, we can
compute the sum of $Y_1$ and $Y_2$
in $O(\log m)$ time using $O(m/\log m)$ processors
in the EREW PRAM model.
\end{lemma}
\begin{proof}
The algorithm follows from the method described above combined
with an EREW PRAM method for merging two sorted lists.
(E.g., see~\cite{DBLP:books/aw/JaJa92} for details on such computations.)
\end{proof}

Given this tool, we now describe our 
condition-number sensitive parallel algorithm for summing
the floating-point numbers in $X$ in parallel.
Our method runs in polylogarithmic time using $O(n\log C(X))$ 
work.\footnote{Technically, we could have $C(X)=1$; 
  hence we could have $\log C(X)=0$. As a notational convenience, 
  therefore, we assume that the
  logarithms in our complexity bounds 
  are defined to always have a minimum value of $1$.}

We begin with a simplified version of our algorithm 
from the previous section.
Initialize an $O(\log n)$-depth summation tree, $T$, to have an element of
$X$ associated with each of its leaves.
Convert each value, $x_i$, stored in a leaf 
to an equivalent sparse $(\alpha,\beta)$-regularized
superaccumulator.
Perform a bottom-up parallel summation computation on $T$ using the method of
Lemma~\ref{lem:merge} to perform each pairwise summation of two sparse
superaccumulators.
We then complete the algorithm by a parallel prefix computation on the sparse
superaccumulator for the root of $T$ to propagate all the 
signed carry bits and we then convert this result
to a floating-point number, as in the last two steps
of our algorithm from the previous section.  
This simplified summation algorithm
runs in
$O(\log^2 n)$ time using $O(n\log n)$ work in the EREW PRAM model.

Given this template, our condition-number sensitive 
algorithm is as follows.
Begin by setting $r=2$.
Then perform the above bottom-up parallel 
merge-and-sum algorithm, but do so using
$r$-truncated sparse superaccumulators.
Unlike our previous method, this one may cause lossy errors to occur, due to
the restrictions to $r$-truncated sparse superaccumulators.
So we utilize a test, which we call the ``stopping condition,''
to determine if the result is correct and we can stop.
If the stopping condition is not satisfied, however, then we 
set $r\leftarrow r^2$ and repeat the computation.
We continue this iteration until either the stopping condition is satisfied
or we have increased $r$ so high that the final sparse superaccumulator is no
longer truncated.

Before we analyze this algorithm, let us first
provide the details for some sufficient stopping conditions.
Let $y$ denote the summation value computed from our method after a given 
iteration, where the final truncated
sparse superaccumulator is 
\[
Y=(y_{i_1},y_{i_2},\ldots,y_{i_r}),
\]
so that $y_{i_r}$ is its least-significant component.
Also, for the sake of simplicity, let us assume that $y$ is positive; the
method for the case when $y$ is negative is similar.
Let $E_{i_r}$ denote the exponent for $y_{i_r}$ 
and let
\[
\epsilon_{\rm min} = \epsilon \times 2^{E_{i_r}},
\]
where $\epsilon$ is the smallest mantissa that can possibly be represented 
in our chosen floating-point representation.
Note that $\epsilon_{\rm min}$
is the smallest value that could be represented by $y_{i_r}$.
A sufficient stopping condition, then, is to test whether or not
\[
y = y \oplus n\epsilon_{\rm min} = y \oplus -n\epsilon_{\rm min}.
\]
that is, the summation value, $y$, would be unchanged even after 
doing a floating-point addition or subtraction
of $n$ copies of a bound whose magnitude is larger than any value we truncated.

As a simplified alternative, which also works in our algorithms
as a sufficient stopping condition, is to
determine whether the exponent for the least
significant bit in $y$ is at least $\lceil \log n\rceil$ greater
than $E_{i_r}$.

The reason that these tests are sufficient as stopping
conditions is that when we are summing the $n$ floating point numbers using
truncated sparse superaccumulators, the largest magnitude
that we can possibility 
truncate any summand is strictly less than $\epsilon_{\rm min}$.
The reason for this is that,
by the definition of a truncated sparse superaccumulator,
if $y_{i_r}$ is included in our final truncated sparse superaccumulator,
then $y_{i_r}$ was not truncated by any partial sum.
Thus, the maximum value of the sum of all the truncated values is at most 
\[
n \epsilon_{\rm min} \le 2^{\lceil \log n\rceil} \epsilon_{\rm min}.
\]
Interestingly, this algorithm gives us the following condition-number
sensitive result.

\begin{theorem}
\label{thm:super-complex}
Given a set, $X$, of $n$ floating-point numbers, we can compute a
faithfully-rounded representation of the sum of the numbers in $X$ in
time
$O(\log^2 n\log\log\log C(X))$ 
using work that is $O(n \log C(X))$ in the EREW PRAM model,
where $C(X)$ is the condition number for $X$.
\end{theorem}
\begin{proof}
The correctness for our algorithm follows immediately from
the above discussion, since we terminate when we are assured of a
faithfully-rounded sum for the numbers in $X$.
We claim that the number of iterations (each of which involves squaring the
truncation parameter, $r$) is $O(\log C(X))$.
To see that this claim is true, note that 
\[
\log C(X) = \log \left(\sum_{i=1}^n |x_i|\right)\ -\ 
            \log \left(\left|\sum_{i=1}^n x_i \right|\right).
\]
Thus, if we represented the values, 
$X_1=\sum |x_i|$ 
and $X_2=|\sum x_i|$, using exact 
fixed-point representations, then
$\log C(X)$ is proportional to
the difference, $\delta$, between the bit-position of the 
most significant $1$-bit in the representation of $X_1$
and the bit-position of the 
most significant $1$-bit in the representation of $X_2$.
Our algorithm must therefore perform a sufficient number of iterations
so that the number of bits in our truncated sparse superaccumulator
for the sum is at least $\delta$.
This indeed occurs in our algorithm and,
furthermore, note that our algorithm will terminate when the number of bits
in our truncated sparse superaccumulator for the sum is $\Theta(\delta+\log n)$.
That is, our algorithm
terminates when $r$ is $O(\log C(X))$, since we assume that 
floating-point numbers in our representation contain 
$\Omega(\log n)$ bits.
Since we square $r$ in each iteration, this implies the claimed 
running-time and work bounds for our parallel algorithm, since we require
$O(\log\log\log C(X)))$ squarings to get $r$ to be large enough and
the total work involved is a geometric summation that adds up to 
$O(n\log C(X))$.
\end{proof}

\ifFull
Thus, for the vast majority of inputs, which have 
constant-bounded condition numbers, our algorithm uses a linear amount of
work and runs in $O(\log^2 n)$ parallel time.
That is, implemented as a sequential algorithm, our method would match the
linear running time of the inherently sequential method of adding $n$
numbers, one at a time to a superaccumulator.
\fi

\section{External-Memory Algorithms}
In this section, we describe our efficient algorithms
for summing $n$ floating-point numbers in the 
external-memory model~\cite{Vitter:2008}.

Suppose we are given a set $X$ of $n$ floating-point numbers.
Our sorting-based external-memory algorithm is as follows.
\begin{enumerate}
\item
Convert each floating-point number to an 
$(\alpha,\beta)$-regularized sparse superaccumulator. This can be
done with a single scan over the input, using 
$O({\rm scan}(n))$ I/Os.
\item
Sort the components of all the sparse superaccumulators
constructed in the previous step
independently by their exponents.
This step clearly takes $O({\rm sort}(n))$ I/Os.
\item
Scan the sorted list of superaccumulator components, while
maintaining a $(\alpha,\beta)$-regularized sparse superaccumulator, $S$, to
hold the sum.  
With each component, $y_{i,j}$, we add $y_{i,j}$ to $S$, using a
localized version of the algorithm for summing two superaccumulators.
Note that we do not need to store all of $S$ in internal memory to
implement this step, however,
since we are processing the components in order by their exponents.
Instead, we just need to keep a hot-swap buffer of $S$ that includes the
current exponent, swapping out blocks to external memory as they become full. 
Moreover, since summing two $(\alpha,\beta)$-regularized
superaccumulators is a carry-free operation, we don't need to worry
about extra I/Os 
that would have otherwise been caused by carry-bit propagation.
Thus, we can implement this step in $O({\rm scan}(n))$ I/Os.
\item
Given the computed superaccumulator, $S$, which now holds the exact
sum of the $n$ floating-point numbers, perform a
back-to-front scan of $S$ to propagate signed carry bits, to convert
$S$ into a non-overlapping $((R/2)-1,(R/2)-1)$-regularized sparse
superaccumulator.
This step clearly requires $O({\rm scan}(n))$ I/Os.
\item
Given a non-overlapping superaccumulator for the sum, we round the
most significant non-zero components to produce a correctly rounded
floating-point number for the sum of the $n$ floating-point numbers
in $X$.
\end{enumerate}

This gives us the following.

\begin{theorem}
Given $n$ floating-point numbers, we can compute a correctly-rounded
floating-point representation of their sum using $O({\rm sort}(n))$
I/Os in the external-memory model, in a cache-oblivious manner.
\end{theorem}
\begin{proof}
We have already established the performance bounds. To establish that
this algorithm can be performed in a cache-oblivious manner, we note
that every step involves simple scans over lists of numbers, except
for a sorting step, which itself can be 
done cache obliviously~\cite{f814600}.
\end{proof}

The critical step in the above external-memory algorithm, of course, is the
scan to add each floating-point component to our superaccumulator.
Since these components were sorted by their exponents and our superaccumulator
representation is carry-free, we need only keep a ``sliding window'' of 
$O(1)$ blocks of our superaccumulator in memory as we perform this scan.
Moreover, such a scan is cache-oblivious given any reasonable page
eviction strategy.

If the size of our superaccumulator, $\sigma(n)$, is less than the size
of internal memory, $M$, however, then we can utilize an even simpler 
algorithm.
Namely, we can simply keep our entire superaccumulator stored in internal
memory and process the floating-point components in any order to add
each one to the superaccumulator. This observation leads to the following,
then.

\begin{theorem}
Given $n$ floating-point numbers, we can compute a correctly-rounded
floating-point representation of their sum using $O({\rm scan}(n))$
I/Os in the external-memory model, 
in a cache-oblivious manner, if $\sigma(n)\le M$.
\end{theorem}

% Add non-default configuration of Spark

\section{Simple MapReduce Algorithms}
\label{sec:mapreduce}

\begin{figure*}[t]
	\centering
		\includegraphics[width=1.00\textwidth]{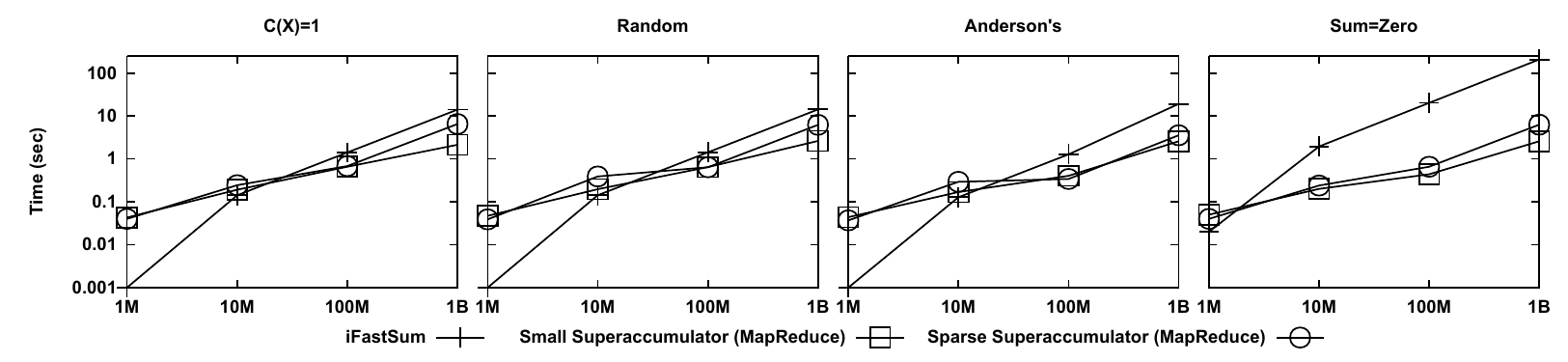}
	\caption{Total running time as the input size increases from 1 million to 1 billion numbers}
	\label{fig:InputSize}
\end{figure*}

In this section, we present a simple 
MapReduce~\cite{Dean:2008,mapsort,Karloff:2010} algorithm for summing $n$
floating-point numbers.
We also report on an experimental evaluation, where
we assume that the input is already loaded in a Hadoop Distributed File System
(HDFS) where the input is partitioned into 128~MB blocks which are stored on the
local disks of cluster nodes.
Our algorithm is based on a single-round of MapReduce which takes as
input a collection of floating point numbers, and produces a single floating
point number that represents the correctly-rounded sum of all input numbers.
In this section, we first give a high level overview of our MapReduce algorithm.
Then, we describe an implementation of that algorithm using Spark~\cite{ZCF+10}.
Finally, we give experimental results of our implementation using large scale data.

\subsection{Overview}
Our MapReduce algorithm runs as a single MapReduce job as detailed below.

\begin{itemize}
	\item Map: for each $x_i$, map $x_i$ to $(r(x_i),x_i)$, where 
	$r(x_i)$ returns an integer in the range $[1,p]$ that represents one of the
 available $p$ reducers. We can simply use a random function $r$, which assigns
 each input record to a randomly chosen reducer. The function $r$ can also be
 defined based on domain knowledge of $X$ with the goal of balancing the load
 across the $p$ reducers.
 For example, if $p$ is $o(n)$, then this function assigns roughly $O(n/p)$
 values for each reducer.

 \item Reduce: In the reduce phase, each reducer $r_i$, $i\in[1,p]$, sums up all
 the assigned numbers using the sequential algorithm described earlier in Section~\ref{sec:pram}.
 The output of the reduce function is one sparse superaccumulator that represents
 the exact sum of the floating-point numbers assigned to $r_i$.
 After that, each reducer writes the resulting sparse superaccumulator to the
 output.
 
 \item Post-process: In this final step, a single machine reads back the $p$
 sparse superaccumulators produced by the $p$ reducers and performs a final
 step of the sparse superaccumulator addition algorithm to add all of them
 into one final sparse superaccumulator. Finally, the resulting sparse superaccumulator
 is converted to a correctly-rounded floating point value which is written
 to the output as the final answer.
\end{itemize}

%This gives us the following.

%\begin{theorem}
%Given $n$ floating-point numbers, we can compute a correctly-rounded
%representation of their sum using a single-round MapReduce algorithm
%that runs in $O(\sigma(n/p)p+n/p)$ time, 
%with high probability, where $p=o(n)$ is the
%number of processors and $\sigma(n)$ is the size of a sparse
%superaccumulator for exactly representing the sum of $n$ floating-point
%numbers.
%\end{theorem}

\subsection{Implementation}
We implemented the MapReduce algorithm described above using
Spark~\cite{ZCF+10}, a modern distributed processing framework.
This implementation is open source and available at https://github.com/aseldawy/sumn.
We begin by loading the input from disk into the distributed memory
of the cluster nodes. In this step, each machine loads the HDFS
blocks that are physically stored on its local disk. Then, each
machine applies a {\em combine} function on each block, which uses
the sequential algorithm described in Section~\ref{sec:pram} to sum
all numbers in each partition into one sparse superaccumulator. The
goal of the combine step is to reduce the size of the data that
need to be shuffled between mappers and reducers. The output of the
combine function is a single key-value pair where the key is a
random integer in the range $[1, p]$, and the value is the sparse
superaccumulator.  Then, the shuffle phase groups key-value pairs
by the reducer number and assigns each group to one reducer.  After
that, each reducer sums up all the assigned sparse superaccumulators
into one sparse superaccumulator and writes it to the intermediate
output. Finally, the postprocess phase runs on the driver machine
that issued the MapReduce job and it combines all sparse superaccumulators
into one and writes the correctly-rounded floating point value as
the final result.

\subsection{Experimental Results}

\begin{figure*}[t]
	\centering
		\includegraphics[width=1.00\textwidth]{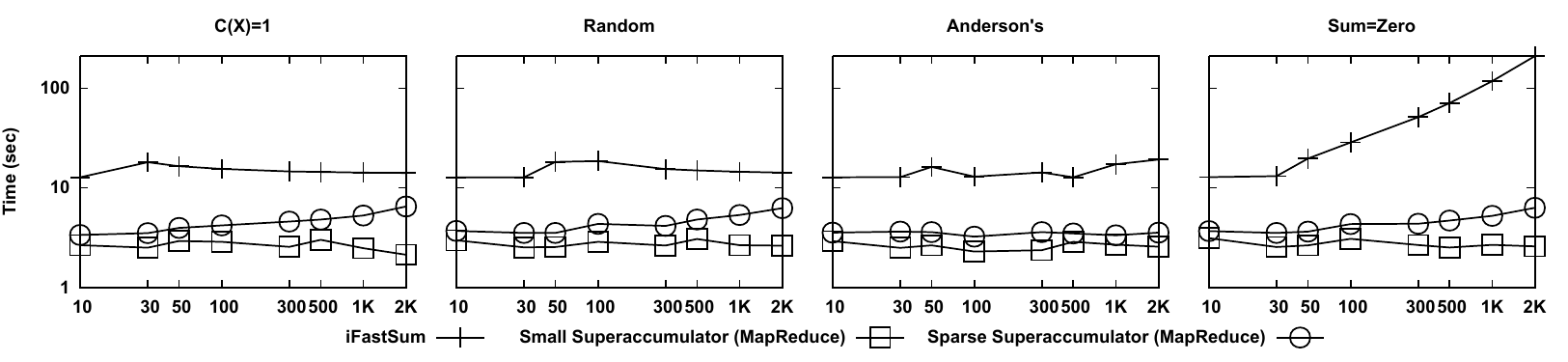}
	\caption{Total running time as the parameter $\delta$ increases from 10 to 2000}
	\label{fig:Delta}
\end{figure*}

\begin{figure*}[t]
	\centering
		\includegraphics[width=1.00\textwidth]{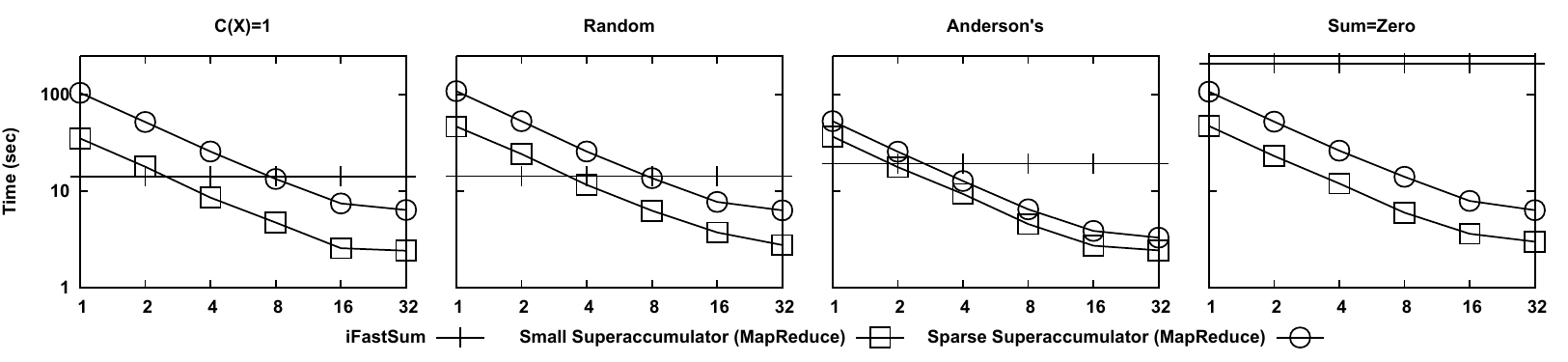}
	\caption{Total running time as the cluster size increases from 1 to 32 cores}
	\label{fig:ClusterSize}
\end{figure*}

To test our MapReduce implementations, we carried out an experimental
evaluation of the proposed MapReduce algorithm on large datasets
of up to one~billion numbers. The input datasets were randomly
generated using four different distributions as described
in~\cite{doi:10.1137/070710020}:
\begin{enumerate}
\item  A set of randomly generated
{\em positive} numbers which results in a condition number $C(X)=1$.
\item
A mix of positive and negative numbers generated uniformly at
random. 
\item
Anderson's-ill conditioned data where a set of random
numbers are generated, and then their arithmetic mean is subtracted
from each number. 
\item
A set of numbers with real sum equal to zero.
\end{enumerate}
In the four datasets, the parameter $\delta$ defines an upper bound
for the range of exponents of input numbers.

We used Apache Spark~1.6.0 running on Ubuntu 14.04 and OpenJDK
1.7.0\_91. 
For experimental repeatability, we ran all our experiments
on Amazon EC2 using an instance of type `r3.8xlarge' which has 32
cores and 244~GB of main memory running on Intel Xeon E5-2670 v2
processors~\cite{xeonE5-2670}. We used the 32~cores to run a local
cluster of 32~worker nodes inside that machine. We ran the experiments
using three algorithms and measured the end-to-end running time of
each one:
\begin{enumerate}
\item Our proposed MapReduce algorithm that uses sparse
superaccumulators. 
\item A variation of our MapReduce algorithm that
uses small superaccumulators~\cite{Neal15a} instead of sparse
superaccumulators. 
\item The state-of-the-art sequential iFastSum
algorithm~\cite{Zhu:2010:A9O}. 
\end{enumerate}
For the latter algorithm,
we used the original C++ implementation
provided by Zhu and Hayes~\cite{Zhu:2010:A9O}, compiled using gcc
4.8.4. For all the techniques, we first generated a dataset using
the random generator provided in~\cite{Zhu:2010:A9O} and stored it
to disk. Then, we processed the same generated dataset with each
algorithm one after another. We ignored the disk I/O time and focused
only on the processing time. If we take the disk I/O into account,
all MapReduce algorithms will be much faster due to the distributed
nature of Hadoop Distributed File System (HDFS) where the machines
load data in parallel from multiple disks.

Figure~\ref{fig:InputSize} shows the overall running time of the
three algorithms as the input size increases from 1~million to
1~billion numbers, while the value of $\delta$ is fixed at 2000.
In general, iFastSum is faster for processing small datasets with
less than 10~million records. However, as the size of the data
increases, both MapReduce implementations outperform the sequential
iFastSum algorithm with up to 80x speedup. This shows a great
potential for MapReduce algorithms in the problem of summing a huge
array of floating point numbers. We observe that the implementation
that uses small superaccumulator is slightly faster than the one
that uses sparse superaccumulator. The reason for this is that each
sparse superaccumulator runs on a single core, which does not allow
it to realize the theoretical limit of doing $O(p)$ work in $O(1)$
time using $p$ processors. In the future, we could investigate the
use of single-instruction multiple-data (SIMD) features to achieve
a higher level of parallelism. Another possibility is to use GPU
processing units which can achieve massive parallelization with
thousands of threads on a single chip. We believe that there is a
huge potential in these two options as the design of sparse
superaccumulator lends itself to these two types of parallelization
where the same instruction is repeated for every index in the
superaccumulator.

Figure~\ref{fig:Delta} shows the running time when the parameter
$\delta$ increases from 10 to 2000. Notice that the maximum possible
value for $\delta$ is 2046 for double-precision floating point
numbers. As the input size is 1 billion, we observe that the two
MapReduce algorithms consistently outperform the sequential iFastSum
algorithm. We also observe that the running time of the sparse
superaccumulator algorithm slightly increases as the value of
$\delta$ increases. This behavior is expected, because the increased
value of $\delta$ makes the superaccumulator less sparse as the
number of non-zero (active) indices increases. The only exception
is with dataset No.~3, as the subtraction of the mean causes the
range of exponents to decrease to about 15 even if the original
$\delta$ is very large.  Similar to the behavior in its original
paper~\cite{doi:10.1137/070710020}, the running time of iFastSum
with dataset No.~4 increases with the value of $\delta$. Interestingly,
the small superaccumulator keeps a consistent performance regardless
of the value of $\delta$.

Figure~\ref{fig:ClusterSize} shows the end-to-end running time as
the cluster size increases from 1 to 32~cores. The performance of
iFastSum stays constant as it runs only on a single core. The results
of this experiment shows a perfect scalability for our MapReduce
algorithm, where the performance scales linearly with number of
machines. The performance starts to saturate as the cluster size
increases from 16 to 32~cores because the underlying processor 
{\em virtually} increases the number of running threads using the
hyper-threading feature. For a small cluster with a few cores,
iFastSum is faster, as it is highly tuned for sequential processing
while Spark incurs extra overhead for the other algorithms. As the
number of cores increases, the parallelization of MapReduce algorithms
allows them to outperform iFastSum in all experimented datasets.
However, the crossover point changes from one dataset to another.
For example, since dataset No.~4 is the worst case for iFastSum,
sparse superaccumulator proves to be faster even on a single core.

\section{Conclusion}
In this paper, we have given a number of efficient parallel algorithms
for computing a faithfully rounded floating-point representation of the
sum of $n$ floating-point numbers.  Our algorithms are designed for
a number of parallel models, including the PRAM, external-memory,
and MapReduce models. The primary design paradigm of our methods
is that of
converting the input values to an intermediate representation, called
a sparse superaccumulator, summing the values exactly in this representation,
and then converting this exact sum to a faithfully-rounded floating-point
representation.
\ifFull
We are able to achieve significant parallelization by utilizing a novel
intermediate floating-point superaccumulator representation that is carry-free.
\fi
Our experimental evaluation shows that our MapReduce algorithm can achieve
up to 80X performance speedup as compared to the state-of-the-art sequential algorithm.
The MapReduce algorithm yields lineary scalability with both the input dataset and
number of cores in the cluster.

\subsection*{Acknowledgments}
This research was supported in part by
the National Science Foundation under grant 1228639,
and an AWS in Education Grant.
We would like to thank Wayne Hayes for several helpful discussions
concerning the topics of this paper.

{\raggedright 
\small
\bibliographystyle{abbrv} 
\bibliography{refs} 
}

\end{document}